\documentclass[sigconf]{aamas} 
\usepackage{balance} % for balancing columns on the final page
\usepackage{multirow}
\usepackage[utf8]{inputenc} % allow utf-8 input
\usepackage[T1]{fontenc}    % use 8-bit T1 fonts
\usepackage{hyperref}       % hyperlinks
\usepackage{url}            % simple URL typesetting
\usepackage{booktabs}       % professional-quality tables
\usepackage{amsfonts}       % blackboard math symbols
\usepackage{nicefrac}       % compact symbols for 1/2, etc.
\usepackage{color}
\usepackage{times}
\usepackage{graphicx} % more modern
\usepackage{subcaption}
\usepackage{natbib}
\usepackage{tikz}
\usepackage{algorithm}
\usepackage{algorithmic}
\usepackage{hyperref}
\usepackage{multicol}
\usepackage{paralist}
\usepackage{array}
\usepackage{geometry}
\usepackage{amsmath,amsthm,mathrsfs,dsfont}
\usepackage{cleveref}

\usepackage{thmtools,thm-restate}
\usepackage{enumitem}

\setcopyright{ifaamas}
\acmConference[AAMAS '22]{In Submission to the 21st International Conference
on Autonomous Agents and Multiagent Systems (AAMAS 2022)}{May 9--13, 2022}
{Auckland, New Zealand}{P.~Faliszewski, V.~Mascardi, C.~Pelachaud,
M.E.~Taylor (eds.)}
\copyrightyear{2022}
\acmYear{2022}
\acmDOI{}
\acmPrice{}
\acmISBN{}

\acmSubmissionID{663}

\title[]{Game Redesign in No-regret Game Playing}

\author{Yuzhe Ma}
\affiliation{
  \institution{University of Wisconsin--Madison}
    \city{}
  \country{}}
\email{yzm234@wisc.edu}

\author{Young Wu}
\affiliation{
  \institution{University of Wisconsin--Madison}
    \city{}
  \country{}}
\email{yw@cs.wisc.edu}

\author{Xiaojin Zhu}
\affiliation{
  \institution{University of Wisconsin--Madison}
    \city{}
  \country{}}
\email{jerryzhu@cs.wisc.edu}

\begin{abstract}
We study the game redesign problem in which an external designer has the ability to change the payoff function in each round, but incurs a design cost for deviating from the original game. 
The players apply no-regret learning algorithms to repeatedly play the changed games with limited feedback. 
The goals of the designer are to 
(i) incentivize all players to take a specific target action profile frequently; 
and (ii) incur small cumulative design cost. 
We present game redesign algorithms with the guarantee that the target action profile is played in $T-o(T)$ rounds while incurring only $o(T)$ cumulative design cost. 
Game redesign describes both positive and negative applications: a benevolent designer who incentivizes players to take a target action profile with better social welfare compared to the solution of the original game, 
or a malicious attacker whose target action profile benefits themselves but not the players.
Simulations on four classic games confirm the effectiveness of our proposed redesign algorithms.
\end{abstract}

\keywords{Game Redesign, No-regret Learning, Target Action Profile, Sublinear Cumulative Design Cost}

\newcommand{\BibTeX}{\rm B\kern-.05em{\sc i\kern-.025em b}\kern-.08em\TeX}

\newcommand\ind[1]{\ensuremath{\mathds{1}\left[#1\right]}}
\newcommand{\setfont}[1]{\mathbb{#1}}
\def\argmin{\ensuremath{\mbox{argmin}}}

\renewcommand{\L}{\mathcal{L}}
\newcommand{\A}{\mathcal{A}}

\newcommand{\E}[2]{\mathbf{E}_{#1}\left[#2\right]}

\newcommand{\R}{\setfont{R}}

\newcommand{\thmref}[1]{Theorem~\ref{#1}}
\newcommand{\lemref}[1]{Lemma~\ref{#1}}

\newtheorem{theorem}{Theorem}

\newtheorem{corollary}[theorem]{Corollary}
\newtheorem{assumption}[theorem]{Assumption}

\newtheorem{definition}{Definition}

\newtheorem*{remark}{Remark}

\begin{document}

%%% The following commands remove the headers in your paper. For final 
%%% papers, these will be inserted during the pagination process.

\pagestyle{fancy}
\fancyhead{}

%%% The next command prints the information defined in the preamble.

\maketitle 

%%%%%%%%%%%%%%%%%%%%%%%%%%%%%%%%%%%%%%%%%%%%%%%%%%%%%%%%%%%%%%%%%%%%%%%%

\section{Introduction}

Consider a normal-form game with loss function $\ell^o$.
This is the ``original game.''
As an example, the volunteer's dilemma (see Table~\ref{tab:VD}) has each player choose whether or not to volunteer for a cause that benefits all players. 
It is known that all pure Nash equilibria in this game involve a subset of the players free-riding the contribution from the remaining players. 
$M$ players, who initially do not know $\ell^o$, use no-regret algorithms to individually choose their action in each of the $t=1 \ldots T$ rounds. 
The players receive limited feedback:
suppose the chosen action profile in round $t$ is $a^t=(a^t_1, \ldots, a^t_M)$, 
then the $i$-th player only receives her own loss $\ell^o_i(a^t)$ and does not observe the other players' actions or losses.

Game redesign is the following task.
A game designer -- not a player -- does not like the solution to $\ell^o$. 
Instead, the designer wants to incentivize a particular target action profile $a^\dagger$, for example ``every player volunteers''.
The designer has the power to redesign the game: before each round $t$ is played, the designer can change $\ell^o$ to some $\ell^t$.  The players will receive the new losses $\ell^t_i(a^t)$, but the designer pays a design cost $C(\ell^o, \ell^t, a^t)$ for that round for deviating from $\ell^o$.
The designer's goal is to make the players play the target action profile $a^\dagger$ in the vast majority ($T-o(T)$) of rounds, while the designer only pays $o(T)$ cumulative design cost.
Game redesign naturally emerges under two opposing contexts:
\begin{itemize}
\item A benevolent designer wants to redesign the game to improve social welfare, as in the volunteer's dilemma;
\item A malicious designer wants to poison the payoffs to force a nefarious target action profile upon the players.  This is an extension of reward-poisoning adversarial attacks (previously studied on bandits~\citep{jun2018adversarial,liu2019data,ma2018data,ming2020attack,guan2020robust, garcelon2020adversarial,bogunovic2021stochastic,zuo2020near,lu2021stochastic} and reinforcement learning~\citep{zhang2020adaptive,ma2019policy,rakhsha2020policy,sun2020vulnerability,huang2019deceptive}) to game playing.
\end{itemize}
For both contexts the mathematical question is the same.  Since the design costs are measured by deviations from the original game $\ell^o$, the designer is not totally free in creating new games. 
Intuitively, the following considerations are sufficient for successful game redesign:
\begin{enumerate}[leftmargin=*]
\item Do not change the loss of the target action profile, i.e. let $\ell^t(a^\dagger)=\ell^o(a^\dagger), \forall t$.
If game redesign is indeed successful, then $a^\dagger$ will be played for $T-o(T)$ rounds.
As we will see, $\ell^t(a^\dagger)=\ell^o(a^\dagger)$ means there is no design cost in those rounds under our definition of $C$.
The remaining rounds incur at most $o(T)$ cumulative design cost. 

\item The target action profile $a^\dagger$ forms a strictly dominant strategy equilibrium. This ensures no-regret players will eventually learn to prefer $a^\dagger$ over any other action profiles.

%\item
% The loss of the redesigned game should only take values allowed by the original game. This first consideration is important because using a redesigned loss not allowed by the original game may change the semantics of the game, resulting in complaint on the designer. 

\end{enumerate}
We formalize these intuitions in the rest of the paper.

\section{The Game Redesign Problem}
We first describe the original game without the designer.
There are $M$ players.
Let $\A_i$ be the finite action space of player $i$, and let $A_i=|\A_i|$.
The original game is defined by the loss function $\ell^o: \A_1 \times \ldots \A_M \mapsto \R^M$.
The players do not know $\ell^o$.  Instead, we assume they play the game for $T$ rounds using no-regret algorithms. 
This may be the case, for example, if the players are learning an approximate Nash equilibrium in zero-sum $\ell^o$ or coarse correlated equilibrium in general sum $\ell^o$.
In running the no-regret algorithm, the players maintain their own action selection policies $\pi_i^t\in \Delta^{\A_i}$ over time, where $\Delta^{\A_i}$ is the probability simplex over $\A_i$. In each round $t$, every player $i$ samples an action $a_i^t$ according to policy $\pi_i^t$.
This forms an action profile $a^t=(a^t_1, \ldots, a^t_M)$.
The original game produces the loss vector $\ell^o(a^t)=(\ell_1^o(a^t),..., \ell_M^o(a^t))$.  
However, player $i$ only observes her own loss value $\ell_i^o(a^t)$, not the other players' losses or their actions.  
All players then update their policy according to their no-regret algorithms.

We now bring in the designer.
The designer knows $\ell^o$ and wants players to frequently play an arbitrary but fixed target action profile $a^\dagger$.
At the beginning of round $t$, the designer commits to a potentially different loss function $\ell^t$.
Note this involves preparing the loss vector $\ell^t(a)$ for all action profiles $a$ (i.e. ``cells'' in the payoff matrix).
The players then choose their action profile $a^t$.
Importantly, the players receive losses $\ell^t(a^t)$, not $\ell^o(a^t)$.
For example, in games involving money such as the volunteer game, the designer may achieve $\ell^t(a^t)$ via taxes or subsidies, and in zero-sum games such as the rock-paper-scissors game, the designer essentially ``makes up'' a new outcome and tell each player whether they win, tie, or lose via $\ell^t_i(a^t)$;
The designer incurs a cost $C(\ell^o, \ell^t, a^t)$ for deviating from $\ell^o$.
The interaction among the designer and the players is summarized as below.
\begin{algorithm}
{Designer knows $\ell^o$, $a^\dagger$, $M$, $\A_i,...,\A_M$, 
and player no-regret rate $\alpha$}
\caption*{\textbf{Protocol}: Game Redesign}
\label{prot:protocol}
\begin{algorithmic}
    \FOR{$t=1,2,\ldots, T$}
    \STATE Designer prepares new loss function $\ell^t$.\label{protocol:design}
    \STATE Players form action profile $a^t=(a_1^t,...,a_M^t)$, where $a_i^t\sim\pi_i^t, \forall i\in [M]$.
    \STATE Player $i$ observes the new loss $\ell_i^t(a^t)$ and updates policy $\pi_i^t$.
    \STATE Designer incurs cost $C(\ell^o, \ell^t, a^t)$.\label{perturbation_implement}
    \ENDFOR
  \end{algorithmic}
\end{algorithm}

The designer has two goals simultaneously: 
\begin{enumerate}[leftmargin=*]
\item 
To incentivize the players to frequently choose the target action profile $a^\dagger$ (which may not coincide with any solution of $\ell^o$).  Let $N^T(a)=\sum_{t=1}^T \ind{a^t=a}$ be the number of times an action profile $a$ is chosen in $T$ rounds, then this goal is to achieve $\E{}{N^T(a^\dagger)}=T-o(T)$. 

\item To have a small cumulative design cost $C^T := \sum_{t=1}^T C(\ell^o, \ell^t, a^t)$, specifically $\E{}{C^T} = o(T)$.
\end{enumerate}

The per-round design cost $C(\ell^o, \ell^t, a)$ is application dependent.
One plausible cost is to account for the ``proposed changes'' in all action profiles, not just what is actually chosen: an example is $C(\ell^o, \ell^t, a^t)=\sum_a \|\ell^o(a)-\ell^t(a)\|$. Note that it ignores the $a^t$ argument.
In many applications, though, only the chosen action profile costs the designer: 
an example is $C(\ell^o, \ell^t, a^t)= \|\ell^o(a^t)-\ell^t(a^t)\|$.
This paper uses a slight generalization of the latter cost:
\begin{assumption}
The non-negative designer cost function $C$ satisfies
$\forall t, \forall a^t, C(\ell^o, \ell^t, a^t)\le \eta\|\ell^o(a^t)-\ell^t(a^t)\|_p$
for some Lipschitz constant $\eta$ and norm $p\ge 1$.
\end{assumption}
This implies no design cost if the losses are not modified, i.e., when $\ell^o(a^t)=\ell^t(a^t)$, $C(\ell^o, \ell^t, a^t)=0$ .

\section{Assumptions on the Players: No-Regret Learning}
\label{sec:no-regret-learning}
The designer assumes that the players 
are each running a no-regret learning algorithm like EXP3.P~\citep{bubeck2012regret}. It is well-known that for two-player ($M=2$) zero-sum games, no-regret learners could approximate an Nash Equilibrium~\citep{blumlearning}. More general results suggest that for multi-player ($M\ge 2$) general-sum games, no-regret learners can approximate a Coarse Correlated Equilibrium~\citep{hart2000simple}. We first define the player's regret. We use $a_{-i}^t$ to denote the actions selected by all players except player $i$ in round $t$.
\begin{definition}{(Regret)}.\label{def:regret} For any player $i$, the best-in-hindsight regret with respect to a sequence of loss functions $\ell_i^t(\cdot, a_{-i}^t), t\in [T]$, is defined as
\begin{equation}\label{eq:regret}
R_i^T=\sum_{t=1}^T \ell_i^t(a_i^t, a_{-i}^t) - \min_{a_i\in \A_i} \sum_{t=1}^T \ell_i^t(a_i, a_{-i}^t).
\end{equation}
The expected regret is defined as $\E{}{R_i^T}$, where the expectation is taken with respect to the randomness in the selection of actions $a^t, t\in [T]$ over all players.
\end{definition}
%A few important remarks are in order.
\begin{remark}
The loss functions $\ell_i^t(\cdot, a_{-i}^t), t\in [T]$ depend on the actions selected by the other players $a_{-i}^t$, while $a_{-i}^t$ further depends on $a^1,...,a^{t-1}$ of all players in the first $t-1$ rounds. Therefore, $\ell_i^t(\cdot, a_{-i}^t)$ depends on $a^1_i,...,a^{t-1}_i$. That means, from player $i$'s perspective, the player is faced with a non-oblivious (adaptive) adversary~\citep{slivkins2019introduction}.
\end{remark}

\begin{remark}
Note that  $a_i^* :=\argmin_{a_i\in \A_i} \sum_{t=1}^T \ell_i^t(a_i, a_{-i}^t)$ in~\eqref{eq:regret} would have meant a baseline in which player $i$ always plays the best-in-hindsight action $a_i^*$ in all rounds $t \in [T]$. Such baseline action should have caused all other players to change their plays away from $a^1_{-i}, ..., a^T_{-i}$.  However, we are disregarding this fact in defining~\eqref{eq:regret} .  For this reason,~\eqref{eq:regret}  is not fully counterfactual, and is called the best-in-hindsight regret in the literature~\citep{bubeck2012regret}.   The same is true when we define expected regret and introduce randomness in players' actions $a^t$.
\end{remark}
Our key assumption is that the learners achieve sublinear expected regret. This assumption is satisfied by standard bandit algorithms such as EXP3.P~\citep{bubeck2012regret}.
\begin{assumption}{(No-regret Learner)}
We assume the players apply no-regret learning algorithm that achieves expected regret $\E{}{R_i^T}=O(T^\alpha), \forall i$ for some $\alpha\in[0, 1)$.
\end{assumption}

%The designer assumes no-regret learning players because it is a standard way for players to achieve approximate solution concepts such as Nash equilibrium or more generally the coarse correlated equilibrium in repeated matrix games. In particular, there are two standard results. We briefly explain these two results in the appendix~\ref{app:convergence} without diving into more details.

\section{Game Redesign Algorithms}
There is an important consideration regarding the allowed values of $\ell^t$.  The original game $\ell^o$ has a set of ``natural loss values'' $\L$.  For example, in the rock-paper-scissors game $\L=\{-1,0,1\}$ for the player wins (recall the value is the loss), ties, and loses, respectively; while for games involving money it is often reasonable to assume $\L$ as some interval $[L, U]$.  Ideally, $\ell^t$ should take values in $\L$ to match the semantics of the game or to avoid suspicion (in the attack context).  Our designer can work with discrete $\L$ (section~\ref{sec:attack_discrete_value}); but for exposition we will first allow $\ell^t$ to take real values in 
$\tilde \L =[L, U]$, where $L=\min_{x\in \L} x$ and $U=\max_{x\in L} x$. We assume $U$ and $L$ are the same for all players and $U>L$, which is satisfied when $\L$ contains at least two distinct values. 

\subsection{Algorithm: Interior Design}
The name refers to the narrow applicability of Algorithm~\ref{alg:interior_design}: 
the original game values for the target action profile
$\ell^o(a^\dagger)$ must all be in the interior of $\tilde \L$.
Formally, we require
$\exists \rho\in (0, \frac{1}{2}(U-L)]$, $\forall i, \ell_i^o(a^\dagger)\in[L+\rho, U-\rho]$. In Algorithm~\ref{alg:interior_design}, we present the interior design.
The key insight of Algorithm~\ref{alg:interior_design} is to keep $\ell^o(a^\dagger)$ unchanged:
If the designer is successful, $a^\dagger$ will be played for $T-o(T)$ rounds.
In these rounds, the designer cost will be zero.  The other $o(T)$ rounds each incur bounded cost.
Overall, this will ensure cumulative design cost $C^T = o(T)$.
For the attack to be successful, the designer can make $a^\dagger$ the strictly dominant strategy in any new games $\ell$.   The designer can do this by judiciously increasing or decreasing the loss of other action profiles in $\ell^o$: there is enough room because $\ell^o(a^\dagger)$ is in the interior. In fact, the designer can design a time-invariant game $\ell^t=\ell$ as Algorithm~\ref{alg:interior_design} shows.

\begin{algorithm}{}
\caption{Interior Design}
\label{alg:interior_design}
\begin{algorithmic}
\REQUIRE the target action profile $a^\dagger$; the original game $\ell^o$.
\ENSURE a time-invariant game $\ell$ constructed as follows:
\STATE \begin{equation}\label{eq:interior_design}
\forall i, a,  \ell_i(a)=\left\{
\begin{array}{ll}
\ell_i^o(a^\dagger)- (1-\frac{d(a)}{M})\rho & \mbox{ if } a_i= a_i^\dagger, \\
\ell_i^o(a^\dagger)+\frac{d(a)}{M}\rho& \mbox{ if } a_i\neq a_i^\dagger,
\end{array}
\right.
\end{equation}
where $d(a)=\sum_{j=1}^M \ind{a_j=a_j^\dagger}$.
  \end{algorithmic}
\end{algorithm}

\begin{restatable}{lemma}{lemPostAttackCost}
\label{lem:post-attack-cost}
The redesigned game~\eqref{eq:interior_design} has the following properties.
\begin{enumerate}[leftmargin=*]
\item $ \forall i, a, \ell_i(a)\in \tilde \L$, thus $\ell$ is valid.\label{lem:post-attack-cost-prop1}
\item For every player $i$, the target action $a_i^\dagger$ strictly dominates any other action by $(1-\frac{1}{M})\rho$, i.e., $\ell_i(a_i, a_{-i})= \ell_i(a_i^\dagger, a_{-i})+(1-\frac{1}{M})\rho, \forall i, a_i\neq a_i^\dagger, a_{-i}$.\label{lem:post-attack-cost-prop2}
\item $\ell(a^\dagger)=\ell^o(a^\dagger)$.\label{lem:post-attack-cost-prop3}
\item If the original loss for the target action profile $\ell^o(a^\dagger)$ is zero-sum, then the redesigned game $\ell$ is also zero-sum.\label{lem:post-attack-cost-prop4}
\end{enumerate}
\end{restatable}
The proof is in appendix.
Our main result is that Algorithm~\ref{alg:interior_design} can achieve $\E{}{N^T(a^\dagger)}=T-O(T^\alpha)$ with 
a small cumulative design cost  $\E{}{C^T} = O(T^\alpha)$.
It is worth noting that even though many entries in the redesigned game $\ell$ can appear to be quite different than the original game $\ell^o$, their contribution to the design cost is small because the design discourages them from being played often.
\begin{restatable}{theorem}{thmAttackVersionOne}
\label{thm:attack_version_01}
A designer that uses Algorithm~\ref{alg:interior_design} can achieve expected number of target plays $\E{}{N^T(a^\dagger)}=T-O(MT^\alpha)$ while incurring expected cumulative design cost $\E{}{C^T}=O(\eta M^{1+\frac{1}{p}}T^\alpha)$.
\end{restatable}
\begin{proof}
Since the designer perturbs $\ell^o(\cdot)$ to $\ell(\cdot)$, the players are equivalently running no-regret algorithms under loss function $\ell$. Note that according to~\lemref{lem:post-attack-cost} property~\ref{lem:post-attack-cost-prop2}, $a_i^\dagger$ is the optimal action for player $i$, and taking a non-target action results in $(1-\frac{1}{M})\rho$ regret regardless of $a_{-i}$, thus the expected regret of player $i$ is
\begin{equation}
\begin{aligned}
&\E{}{R_i^T} = \E{}{\sum_{t=1}^T \ind{a_i^t\neq a_i^\dagger}(1-\frac{1}{M})\rho}\\
&= (1-\frac{1}{M})\rho \left(T-\E{}{ N_i^T(a_i^\dagger)}\right)
\end{aligned}
\end{equation}
Rearranging, we have
\begin{equation}\label{eq:bound_Na_01}
 \forall i, \E{}{N_i^T(a_i^\dagger)}= T-\frac{M}{(M-1)\rho}\E{}{R_i^T}
 \end{equation} 
Applying a union bound over $M$ players,
\begin{equation}\label{eq:bound_Na_02}
\begin{aligned}
&T-\E{}{N^T(a^\dagger)}=\E{}{\sum_{t=1}^T\ind{a^t\neq a^\dagger}}\\
&=\E{}{\sum_{t=1}^T\ind{a_j^t\neq a_j^\dagger \text{ for some } j}}\le \E{}{\sum_{t=1}^T \sum_{j=1}^M \ind{a_j^t\neq a_j^\dagger}}\\
&= \sum_{j=1}^M \E{}{\sum_{t=1}^T\ind{a_j^t\neq a_j^\dagger}}=\sum_{j=1}^M \left(T-\E{}{N_j(a_j^\dagger)}\right)\\
&=\sum_{j=1}^M\frac{M}{(M-1)\rho}\E{}{R_i^T}=\frac{M^2}{(M-1)\rho} O(T^\alpha)=O(MT^\alpha).
\end{aligned}
\end{equation}
where the second-to-last equation is due to the no-regret assumption of the learner.
Therefore, we have $\E{}{N^T(a^\dagger)}= T-O(M T^{\alpha})$.

Next we bound the expected cumulative design cost. Note that by design $\ell^o(a^\dagger)=\ell(a^\dagger)$, thus when $a^t=a^\dagger$ by our assumption on the cost function we have $C(\ell^o, \ell, a^t)=0$.
On the other hand, when $a^t\neq a^\dagger$ by Lipschitz condition on the cost function we have $C(\ell^o, \ell, a^t)\le \eta M^{\frac{1}{p}}(U-L)$.
Therefore, the expected cumulative design cost is
\begin{equation}
\begin{aligned}
&\E{}{C^T}=\E{}{\sum_{t=1}^T C(\ell^o, \ell, a^t)}\\
&\le \eta M^{\frac{1}{p}}(U-L)\E{}{\sum_{t=1}^T\ind{a^t\neq a^\dagger}}\\
&=\eta M^{\frac{1}{p}}(U-L)\left(T-\E{}{N^T(a^\dagger)}\right)\\
&=\eta M^{\frac{1}{p}}(U-L)\sum_{j=1}^M\frac{M}{(M-1)\rho}\E{}{R_i^T}=O(\eta M^{1+\frac{1}{p}}T^\alpha),
\end{aligned}
\end{equation}
where the last equality used~\eqref{eq:bound_Na_02}.
\end{proof}

We have two corollaries from~\thmref{thm:attack_version_01}. First, the standard no-regret algorithm EXP3.P~\citep{bubeck2012regret} achieves $\E{}{R_i^T}=O(T^{\frac{1}{2}})$. Therefore,  by plugging $\alpha=\frac{1}{2}$ into~\thmref{thm:attack_version_01} we have:
\begin{corollary}
If the players use EXP3.P, the designer can achieve expected number of target plays $\E{}{N^T(a^\dagger)}=T-O(MT^\frac{1}{2})$ while incurring expected cumulative design cost $\E{}{C^T}=O(\eta M^{1+\frac{1}{p}}T^{\frac{1}{2}})$.
\end{corollary}
If the original game $\ell^o$ is two-player zero-sum, then the designer can also make the players think that $a^\dagger$ is a pure Nash equilibrium.
\begin{restatable}{corollary}{ColNE}
Assume $M=2$ and the original game $\ell^o$ is zero-sum. Then with the redesigned game $\ell$~\eqref{eq:interior_design}, the expected averaged policy $\E{}{\bar \pi^T_i}=\E{}{\frac{1}{T}\sum_t \pi_i^t}$ converges to a point mass on $a_i^\dagger$.
\end{restatable}
\begin{proof}
The new game $\ell$ is also a two-player zero-sum game.
The players applying no-regret algorithm will have their average actions $\E{}{\bar \pi^T}$ converging to an approximate Nash equilibrium. 
%Next we prove that $\E{}{\bar \pi^T_i}$ converges to a point mass distribution concentrated on $a_i^\dagger$.
We use $\pi_i^t(a)$ to denote the probability of player $i$ choosing action $a$ at round $t$. Next we compute $\E{}{\bar \pi_i^T(a^\dagger)}$. Note that this expectation is with respect to all the randomness during game playing, including the selected actions $a^{1:T}$ and policies $\pi^{1:T}$. For any $t$, when we condition on $\pi^t$, we have $\E{}{\ind{a_i^t=a}\mid \pi^t}=\pi_i^t(a)$. Therefore, we have $\forall i$
\begin{equation}
\begin{aligned}
&\E{}{\bar \pi^T_i(a_i^\dagger)}=\frac{1}{T}\E{}{ \sum_{t=1}^T\pi_i^t(a_i^\dagger)}\\
&= \frac{1}{T}\E{\pi^{1:T}}{\sum_{t=1}^T\E{a^t}{\ind{a_i^t=a_i^\dagger}\mid \pi^t}}\\
&=\frac{1}{T}\E{\pi^{1:T}}{\E{a^{1:T}}{\sum_{t=1}^T\ind{a_i^t=a_i^\dagger}\mid \pi^{1:T}}}\\
&=\frac{1}{T}\E{\pi^{1:T}}{\E{a^{1:T}}{N_i^T(a_i^\dagger)\mid \pi^{1:T}}}\\
&=\frac{1}{T}\E{}{N_i^T(a_i^\dagger)}=\frac{T-O(T^\alpha)}{T}\rightarrow 1.
\end{aligned}
\end{equation}
Therefore, asymptotically the players believe that $a_i^\dagger, i\in [M]$ form a Nash equilibrium.
\end{proof}

\subsection{Boundary Design}
When the original game has some $\ell^o(a^\dagger)$ values hitting the boundary of $\tilde \L$, the designer cannot apply Algorithm~\ref{alg:interior_design} directly because the loss of other action profiles cannot be increased or decreased further to make $a^\dagger$ a dominant strategy. 
However, a time-varying design can still ensure
$\E{}{N^T(a^\dagger)}=T-o(T)$ and $\E{}{C^T} = o(T)$. In Algorithm~\ref{alg:boundary_design}, we present the boundary design which is applicable to both boundary and interior $\ell^o(a^\dagger)$ values.

\begin{algorithm}
\caption{Boundary Design}
\label{alg:boundary_design}
\begin{algorithmic}[1]
\REQUIRE the target action profile $a^\dagger$; a loss vector $v\in \R^M$ whose elements are in the interior, i.e.,
$\forall i, v_i \in [L+\rho, U-\rho]$ for some $\rho>0$; the regret rate $\alpha$; $\epsilon\in (0, 1-\alpha)$; the time step $t$.
\ENSURE a time-varying game with loss $\ell^t$.
\STATE Use $v$ in place of $\ell^o(a^\dagger)$ in~\eqref{eq:interior_design} and apply the interior design~\ref{alg:interior_design}.
Call the resulting time-invariant game the ``source game'' $\underline \ell$. 
\STATE Define a ``destination game'' $\overline \ell$ where $\overline \ell(a)=\ell^o(a^\dagger), \forall a$.
\STATE Interpolate the source and destination games:
\begin{equation}\label{eq:boundary_design}
\ell^t= w_t \underline\ell + (1-w_t)\overline \ell
\end{equation}
where 
\begin{equation}\label{eq:wt}
w_t=t^{\alpha+\epsilon-1}
\end{equation}
  \end{algorithmic}
\end{algorithm}
The designer can choose an arbitrary loss vector $v$ as long as $v$ lies in the interior of $\tilde \L$. We give two exemplary choices of $v$.
\begin{enumerate}[leftmargin=*]
\item Let the average player cost of $a^\dagger$ be $\bar \ell(a^\dagger)=\sum_{i=1}^M\ell_i^o(a^\dagger)/M$, then if $\bar \ell(a^\dagger)\in(L, U)$, one could choose $v$ to be a constant vector with value $\bar \ell(a^\dagger)$. The nice property about this choice is that if $\ell^o$ is zero-sum, then $v$ is zero-sum, thus property~\ref{prop4:TVA} is satisfied and the redesigned game is zero-sum. However, note that when $\bar \ell(a^\dagger)$ does hit the boundary, the designer cannot choose this $v$. 

\item Choose $v$ to be a constant vector with value $(L+U)/2$. This choice is always valid, but may not preserve the zero-sum property of the original game unless $L=-U$.
\end{enumerate}
The designer applies the interior design on $v$ to obtain a ``source game'' $\underline \ell$. Note that the target action profile $a^\dagger$ strictly dominates in the source game. 
The designer also creates a ``destination game'' $\overline \ell(a)$ by repeating the $\ell^o(a^\dagger)$ entry everywhere.
The boundary algorithm then interpolates between the source and destination games with a decaying weight $w_t$.
Note
after interpolation~\eqref{eq:boundary_design}, the target $a^\dagger$ still dominates by roughly $w_t$. We design the weight $w_t$ as in~\eqref{eq:wt} so that cumulatively, the sum of $w_t$ grows with rate $\alpha+\epsilon$, which is faster than the regret rate $\alpha$. This is a critical consideration to enforce frequent play of $a^\dagger$. Also note that asymptotically, $\ell^t$ converges toward the destination game. Therefore, in the long run, when $a^\dagger$ is played the designer incurs diminishing cost, resulting in $o(T)$ cumulative design cost.

\begin{restatable}{lemma}{lemPostAttackCostVersionTwo}
\label{lem:post-attack-cost-version2}
The redesigned game~\eqref{eq:boundary_design} has the following properties.
\begin{enumerate}[leftmargin=*]
\item $\forall i, a, \ell^t_i(a)\in\tilde \L$, thus the loss function is valid.\label{prop1:TVA} 
\item For every player $i$, the target action $a_i^\dagger$ strictly dominates any other action by $(1-\frac{1}{M})\rho w_t$, i.e., $\ell_i^t(a_i, a_{-i})= \ell_i^t(a_i^\dagger, a_{-i})+(1-\frac{1}{M})\rho w_t, \forall i, t, a_i\neq a_i^\dagger,  a_{-i}$.\label{prop2:TVA}
\item $\forall t, C(\ell^o, \ell^t, a^\dagger)\le  \eta \|\ell^o(a^\dagger)-\ell^t(a^\dagger)\|_p w_t\le \eta (U-L)M^{\frac{1}{p}}w_t$\label{prop3:TVA}
\item If the original loss for the target action profile $\ell^o(a^\dagger)$ and the vector $v$ are both zero-sum, then $\forall t, \ell^t$ is zero-sum.\label{prop4:TVA}
\end{enumerate}
\end{restatable}

Given~\lemref{lem:post-attack-cost-version2}, we provide our second main result.
\begin{restatable}{theorem}{thmAttackVersionTwo}
\label{thm:attack_version_2}
$\forall \epsilon\in (0, 1-\alpha]$, a designer that uses Algorithm~\ref{alg:boundary_design} can achieve expected number of target plays $\E{}{N^T(a^\dagger)}=T-O(MT^{1-\epsilon})$ while incurring expected cumulative design cost $\E{}{C^T}=O(M^{1+\frac{1}{p}}T^{1-\epsilon}+M^{\frac{1}{p}}T^{\alpha+\epsilon})$.
\end{restatable}
\begin{remark}
By choosing a larger $\epsilon$ in~\thmref{thm:attack_version_2}, the designer can increase $\E{}{N^T(a^\dagger)}$. However, the cumulative design cost can grow. The design cost attains the minimum order $O\left(M^\frac{1}{p}(1+M)T^{\frac{1+\alpha}{2}}\right)$ when $\epsilon=\frac{1-\alpha}{2}$. The corresponding number of target action selection is $\E{}{N^T(a^\dagger)}=T-O(M T^{\frac{1+\alpha}{2}})$
\end{remark}
\begin{proof}
Under game redesign, the players are equivalently running no-regret algorithms over the game sequence $\ell^1, \ldots, \ell^T$ instead of $\ell^o(\cdot)$. 
By~\lemref{lem:post-attack-cost-version2} property~\ref{prop2:TVA}, $a_i^\dagger$ is always the optimal action for player $i$, and taking a non-target action results in $(1-1/M)\rho w_t$ regret regardless of $a_{-i}$, thus the expected regret of player $i$ is
\begin{equation}
\begin{aligned}
&\E{}{R_i^T} = \E{}{\sum_{t=1}^T \ind{a_i^t\neq a_i^\dagger}(1-\frac{1}{M})\rho w_t}\\
&=(1-\frac{1}{M})\rho \E{}{\sum_{t=1}^T \ind{a_i^t\neq a_i^\dagger} w_t}.
\end{aligned}
\label{eq:ERiT}
\end{equation}
Now note that $ w_t=t^{\alpha+\epsilon-1}$ is monotonically decreasing as $t$ grows, thus we have
\begin{equation}
\begin{aligned}
&\sum_{t=1}^T \ind{a_i^t\neq a_i^\dagger} w_t\ge \sum_{t=N_i(a_i^\dagger)+1}^T t^{\alpha+\epsilon-1}\\
&=\sum_{t=1}^T t^{\alpha+\epsilon-1}-\sum_{t=1}^{N_i(a_i^\dagger)} t^{\alpha+\epsilon-1}.
\end{aligned}
\end{equation}
Next, by examining the area under curve, we obtain
\begin{equation}
\begin{aligned}
\sum_{t=1}^T t^{\alpha+\epsilon-1}&\ge \int_{1}^{T} t^{\alpha+\epsilon-1} dt=\frac{1}{\alpha+\epsilon} T^{\alpha+\epsilon} - {1 \over \alpha+\epsilon}.
\end{aligned}
\end{equation}
Similarly, we can also derive
\begin{equation}
\begin{aligned}
\sum_{t=1}^{N_i(a_i^\dagger)} t^{\alpha+\epsilon-1}&\le \int_{0}^{N_i(a_i^\dagger)}t^{\alpha+\epsilon-1} dt=\frac{1}{\alpha+\epsilon}\left(N_i^T(a_i^\dagger)\right)^{\alpha+\epsilon}.
\end{aligned}
\end{equation}
Therefore, we have
\begin{equation}\label{eq:bound_ind}
\begin{aligned}
&\sum_{t=1}^T \ind{a_i^t\neq a_i^\dagger} w_t \ge \frac{1}{\alpha+\epsilon} \left(T^{\alpha+\epsilon}-\left(N_i^T(a_i^\dagger)\right)^{\alpha+\epsilon}\right) - {1 \over \alpha+\epsilon}\\
&=\frac{1}{\alpha+\epsilon}T^{\alpha+\epsilon}\left(1-(1-\frac{T-N_i^T(a_i^\dagger)  }{T})^{\alpha+\epsilon}\right) - {1 \over \alpha+\epsilon}\\
&\ge \frac{1}{\alpha+\epsilon}T^{\alpha+\epsilon} \frac{T-N_i^T(a_i^\dagger) }{T}(\alpha+\epsilon)- {1 \over \alpha+\epsilon}\\
&=T^{\alpha+\epsilon}-T^{\alpha+\epsilon-1} N_i^T(a_i^\dagger)- {1 \over \alpha+\epsilon}.
\end{aligned}
\end{equation}
The inequality follows from the fact $(1-x)^c \le 1-cx$ for $x,c \in (0,1)$.
Plug back in~\eqref{eq:ERiT} we have
\begin{equation}
\begin{aligned}
&\E{}{R_i^T}=(1-\frac{1}{M})\rho\E{}{\sum_{t=1}^T \ind{a_i^t\neq a_i^\dagger} w_t}\\
&\ge (1-\frac{1}{M})\rho\E{}{\left(T^{\alpha+\epsilon}-T^{\alpha+\epsilon-1} N_i^T(a_i^\dagger)- {1 \over \alpha+\epsilon}\right)}\\
&=(1-\frac{1}{M})\rho\left(T^{\alpha+\epsilon}-T^{\alpha+\epsilon-1}\E{}{N_i^T(a_i^\dagger)}- {1 \over \alpha+\epsilon}\right)
\end{aligned}
\end{equation}
As a result, we have
\begin{equation}
\begin{aligned}
\forall i, &\E{}{N_i^T(a_i^\dagger)}\ge T-\frac{M}{(M-1)\rho}\E{}{R_i^T} T^{1-\alpha-\epsilon}-\frac{1}{\alpha+\epsilon}T^{1-\alpha-\epsilon}\\
&=T-\frac{M}{(M-1)\rho}O(T^\alpha) T^{1-\alpha-\epsilon}-\frac{1}{\alpha+\epsilon}T^{1-\alpha-\epsilon}\\
&=T-O(T^{1-\epsilon})-O(T^{1-\alpha-\epsilon})\\
&=T-O(T^{1-\epsilon}).
\end{aligned}
\end{equation}
By a union bound similar to~\eqref{eq:bound_Na_02}, we have $\E{}{N^T(a^\dagger)}=T-O(MT^{1-\epsilon})$.

We now analyze the cumulative design cost. Note that by~\lemref{lem:post-attack-cost-version2} property~\ref{prop3:TVA}, when $a^t=a^\dagger$, $C(\ell^o, \ell^t, a^t)\le \eta (U-L)M^{\frac{1}{p}} w_t$. On the other hand, when $a^t\neq a^\dagger$, we have 
\begin{equation}
\begin{aligned}
C(\ell^o, \ell^t, a^t)&\le \eta \|\ell^o(a^t)-\ell^t(a^t)\|_p\le \eta(U-L) M^{\frac{1}{p}}.
\end{aligned}
\end{equation}
Therefore, the expected cumulative design cost is
\begin{equation}
\begin{aligned}
&\E{}{C^T} \le \eta(U-L) M^{\frac{1}{p}}\E{}{\sum_{t=1}^T\ind{a^t\neq a^\dagger}}\\
&+\eta (U-L)M^{\frac{1}{p}}\E{}{\sum_{t=1}^T\ind{a^t= a^\dagger} w_t}\\
&\le \eta (U-L)M^{\frac{1}{p}}(T-\E{}{N^T(a^\dagger)})+\eta (U-L)M^{\frac{1}{p}}\sum_{t=1}^T w_t.
\end{aligned}
\end{equation}
$T-\E{}{N^T(a^\dagger)}=O(MT^{1-\epsilon})$ is already proved. Also note that
\begin{equation}
\begin{aligned}
 \sum_{t=1}^T w_t&=\sum_{t=1}^T t^{\alpha+\epsilon-1}\le \int_{t=0}^T  t^{\alpha+\epsilon-1}=\frac{1}{\alpha+\epsilon} T^{\alpha+\epsilon}.
 \end{aligned}
\end{equation}
Therefore, we have
\begin{equation}
\begin{aligned}
\E{}{C^T}&\le (U-L)\eta M^{\frac{1}{p}}O(MT^{1-\epsilon})+\frac{\eta (U-L)}{\alpha+\epsilon} M^{\frac{1}{p}} T^{\alpha+\epsilon}\\
&=O(M^{1+\frac{1}{p}}T^{1-\epsilon}+M^{\frac{1}{p}}T^{\alpha+\epsilon}).
\end{aligned}
\end{equation}
\end{proof}

\begin{corollary}
Assume the no-regret learning algorithm is EXP3.P. Then by picking $\epsilon=\frac{1}{4}$ in~\thmref{thm:attack_version_2}, a designer can achieve expected number of target plays $\E{}{N^T(a^\dagger)}=T-O(MT^{\frac{3}{4}})$ while incurring $\E{}{C^T}=O\left(M^\frac{1}{p}(1+M)T^{\frac{3}{4}}\right)$ design cost.
\end{corollary}

\subsection{Discrete Design}
\label{sec:attack_discrete_value}
In previous sections, we assumed the games $\ell^t$ can take arbitrary continuous values in the relaxed loss range $\tilde \L=[L, U]$. However, there are many real-world situations where continuous loss does not have a natural interpretation. For example, in the rock-paper-scissors game, the loss is interpreted as win, lose or tie, thus $\ell^t$ should only take value in the original loss value set $\L=\{-1, 0, 1\}$. 
We now provide a discrete redesign to convert any game $\ell^t$ with values in $\tilde \L$ into a game $\hat\ell^t$ only involving loss values $L$ and $U$, which are both in $\L$.
Specifically, the discrete design is illustrated in Algorithm~\ref{alg:discrete_design}.
\floatname{algorithm}{Algorithm}
\begin{algorithm}
\caption{Discrete Design}
    \label{alg:discrete_design}
      \begin{algorithmic}
      \REQUIRE the target action profile $a^\dagger$; a loss vector $v\in \R^M$ whose elements are in the interior, i.e.,
$\forall i, v_i \in [L+\rho, U-\rho]$ for some $\rho>0$; the regret rate $\alpha$; $\epsilon\in (0, 1-\alpha)$; the time step $t$.
\ENSURE a time-varying game with loss $\hat \ell^t\in \L$ as below:
    \STATE \begin{equation}\label{eq:discrete_design}
\forall i, a,  \hat \ell_i^t(a)=\left\{
\begin{array}{ll}
U & \mbox{ with probability $\frac{\ell_i^t(a)-L}{U-L}$ } \\
L & \mbox{ with probability $\frac{U-\ell_i^t(a)}{U-L}$}.
\end{array}
\right.
\end{equation}
  \end{algorithmic}
\end{algorithm}

It is easy to verify $\E{}{\hat\ell^t}=\ell^t$.
In experiments we show such discrete games also achieve the design goals.

\subsection{Thresholding the Redesigned Game}
For all designs in previous sections, the designer could impose an additional min or max operator to threshold on the original game loss, e.g., for the interior design, the redesigned game loss after thresholding becomes
\begin{equation}\label{eq:interior_design_minmax}
\forall i, a,  \ell_i(a)=\left\{
\begin{array}{ll}
\min\{\ell_i^o(a^\dagger)- (1-\frac{d(a)}{M})\rho, \ell^o(a)\} & \mbox{ if } a_i= a_i^\dagger, \\
\max\{\ell_i^o(a^\dagger)+\frac{d(a)}{M}\rho, \ell^o(a)\}& \mbox{ if } a_i\neq a_i^\dagger,
\end{array}
\right.
\end{equation}
We point out a few differences between~\eqref{eq:interior_design_minmax} and~\eqref{eq:interior_design}.  First,~\eqref{eq:interior_design_minmax} guarantees a dominance gap of ``at least'' (instead of exactly) $(1-\frac{1}{M})\rho$. As a result, the thresholded game can induce a larger $N^T(a^\dagger)$ because the target action $a^\dagger$ is redesigned to stand out even more.  Second, one can easily show that~\eqref{eq:interior_design_minmax} incurs less design cost $C^T$ compared to~\eqref{eq:interior_design} due to thresholding. \thmref{thm:attack_version_01} still holds.  However, thresholding no longer preserves the zero-sum property~\ref{lem:post-attack-cost-prop4} in~\lemref{lem:post-attack-cost} and~\lemref{lem:post-attack-cost-version2}. When such property is not required, the designer may prefer~\eqref{eq:interior_design_minmax} to slightly  improve the redesign performance. The thresholding also applies to the boundary and discrete designs.

\begin{table*}[t]
%       \begin{minipage}[t]{0.3\linewidth}
        \begin{minipage}{0.47\linewidth} %changed <<<<<<<<<<<<
        \begin{center}
               \begin{tabular}{cc|c|c|}
      & \multicolumn{1}{c}{} & \multicolumn{2}{c}{Other players}\\
      & \multicolumn{1}{c}{} & \multicolumn{1}{c}{exists a volunteer}  & \multicolumn{1}{c}{no volunteer exists} \\\cline{3-4}
      \multirow{2}*{Player $i$}  & volunteer & $0$ & $ 0$ \\\cline{3-4}
      & not volunteer & $-1$ & $10$ \\\cline{3-4}
    \end{tabular}%
    \vspace{0.3cm}
 \caption{The loss function $\ell^o_i$ for individual player $i=1 \ldots M$ in the Volunteer Dilemma.}
\label{tab:VD}
\end{center}
        \end{minipage}% 
        \hfill
    %% remove blank lines <<<<<<<<<<<<<<<<
%       \begin{minipage}[t]{0.7\linewidth}
        \begin{minipage}{0.47\linewidth} %changed <<<<<<<<<<<
                \begin{center}
                \begin{tabular}{cc|c|c|c|}
  \multicolumn{1}{c}{}    & \multicolumn{1}{c}{} & \multicolumn{3}{c}{Number of other volunteers}\\
   \multicolumn{1}{c}{}   & \multicolumn{1}{c}{} & \multicolumn{1}{c}{0}  & \multicolumn{1}{c}{1} & \multicolumn{1}{c}{2}\\\cline{3-5}
       \multirow{2}*{Player $i$}  & volunteer      & $-2/3$ & $-1/3$     & $0$ \\\cline{3-5}
      					     & not volunteer &\hspace {3mm} $10$ \hspace {3mm}             &  \hspace {3mm}$1/3$  \hspace {3mm}      & $2/3$ \\\cline{3-5}
    \end{tabular}
    \vspace{0.3cm}
\caption{The redesigned loss function $\ell_i$ for individual player $i$ in VD.}
\label{tab:VD_redesigned}
\end{center}
\end{minipage}
\vspace{-0.3cm}
    \end{table*} 
    
\section{Experiments}
We perform empirical evaluations of game redesign algorithms on four games --- the volunteer's dilemma (VD), tragedy of the commons (TC), prisoner's dilemma (PD) and rock-paper-scissors (RPS). Throughout the experiments, we use EXP3.P~\citep{bubeck2012regret} as the no-regret learner. The concrete form of the regret bound for EXP3.P is illustrated in the appendix~\ref{sec:exact_form}. Based on that, we derive the exact form of our theoretical upper bounds for~\thmref{thm:attack_version_01} and~\thmref{thm:attack_version_2} (see~\eqref{eq:interior_exact_bound_N}-\eqref{eq:boundary_exact_bound_C}), and we show the theoretical value for comparison in our experiments. We let the designer cost function be $C(\ell^o, \ell^t, a^t)=\|\ell^o(a^t)-\ell^t(a^t)\|_p$ with $p=1$. 
For VD, TC and PD, the original game is not zero-sum, and we apply the thresholding~\eqref{eq:interior_design_minmax} to slightly improve the redesign performance. For the RPS game, we apply the design without thresholding to preserve the zero-sum property. The results we show in all the plots are  produced by taking the average of 5 trials.

\subsection{Volunteer's Dilemma (VD)}\label{sec:VD_game}
In volunteer's dilemma (Table~\ref{tab:VD}) there are $M$ players. 
Each player has two actions: volunteer or not. When there exists at least one volunteer, those players who do not volunteer gain 1 (i.e. a $-1$ loss). The volunteers, however, receive zero payoff. On the other hand, if no players volunteer, then every player suffers a loss of 10.
As mentioned earlier, all pure Nash equilibria involve free-riders.
The designer aims at encouraging all players to volunteer, i.e., the target action profile $a_i^\dagger$ is ``volunteer'' for any player $i$. Note that $\forall i, \ell_i^o(a^\dagger)=0$, which lies in the interior of $\L=[-1, 10]$. Therefore, the designer could apply the interior design Algorithm~\ref{alg:interior_design}. The margin parameter is $\rho=1$. We let $M=3$.
In table~\ref{tab:VD_redesigned}, we show the redesigned game $\ell$. Note that when all three players volunteer (i.e., at $a^\dagger$), the loss is unchanged compared to $\ell^o$. Furthermore, regardless of the other  players, the action ``volunteer'' strictly dominates the action ``not volunteer'' by at least $(1-\frac{1}{M})\rho=\frac{2}{3}$ for every player. When there is no other volunteers, the dominance gap is $\frac{32}{3}\ge(1-\frac{1}{M})\rho$, which is due to the thresholding in~\eqref{eq:interior_design_minmax}.

We simulated play for $T=10^4,10^5,10^6,10^7$, respectively on this redesigned game $\ell$.
In Figure~\ref{fig:VD_TIA}\subref{fig:VD_TIA_N}, we show $T-N^T(a^\dagger)$ against $T$. The plot is in log scale. 
The standard deviation estimated from 5 trials is less than $3\%$ of the corresponding value and is hard to see in log-scale plot, thus we do not show that.
We also plot our theoretical upper bound in dashed lines for comparison. 
Note that the theoretical value indeed upper bounds our empirical results. 
In Figure~\ref{fig:VD_TIA}\subref{fig:VD_TIA_cost}, we show $C^T$ against $T$. 
Again, the theoretical upper bound holds.
As our theory predicts, for the four $T$'s the designer increasingly enforces $a^\dagger$ in 60\%, 82\%, 94\%, and 98\% of the rounds, respectively;  The per-round design costs $C^T/T$ decreases at 0.98, 0.44, 0.15, and 0.05, respectively.

\begin{figure}[t]
\begin{subfigure}{0.48\linewidth}
  \centering
  \includegraphics[width=1\textwidth]{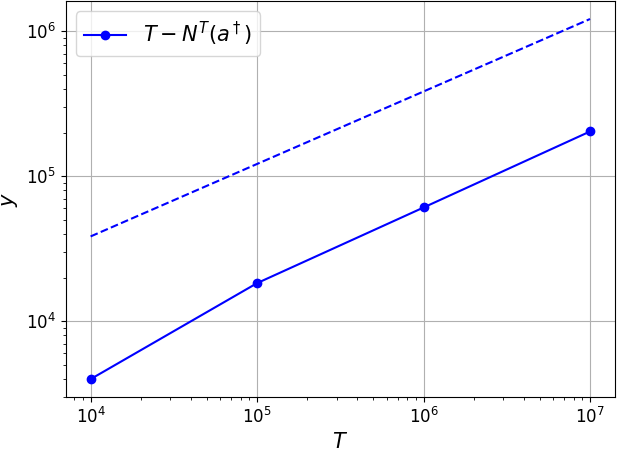}  
  \caption{Number of rounds with $a^t\neq a^\dagger$ grows sublinearly}
  \label{fig:VD_TIA_N}
\end{subfigure}
\hfill
\begin{subfigure}{0.48\linewidth}
  \centering
  \includegraphics[width=1\textwidth]{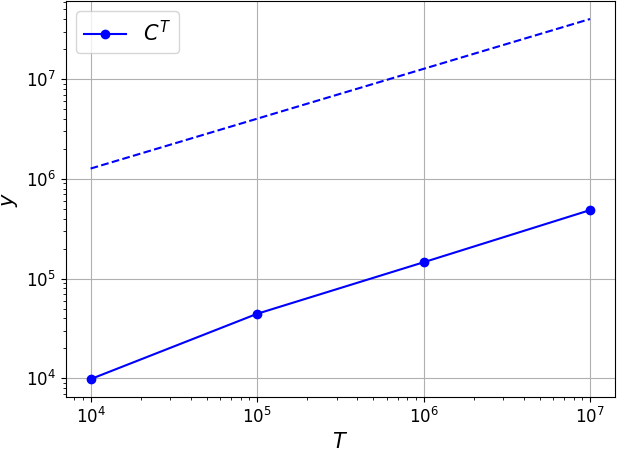}  
  \caption{The cumulative design cost grows sublinearly too}
  \label{fig:VD_TIA_cost}
\end{subfigure}
\caption{Interior design (Algorithm~\ref{alg:interior_design}) on VD with $M=3$. The dashed lines show the slope of sublinear rate $\sqrt{T}$ in log-log scale.}
\label{fig:VD_TIA}
\end{figure}

\begin{figure*}[t!]
\centering
\begin{subfigure}{0.92\textwidth}
\begin{subfigure}{0.325\textwidth}
  \centering
  % include first image
  \includegraphics[width=1\textwidth]{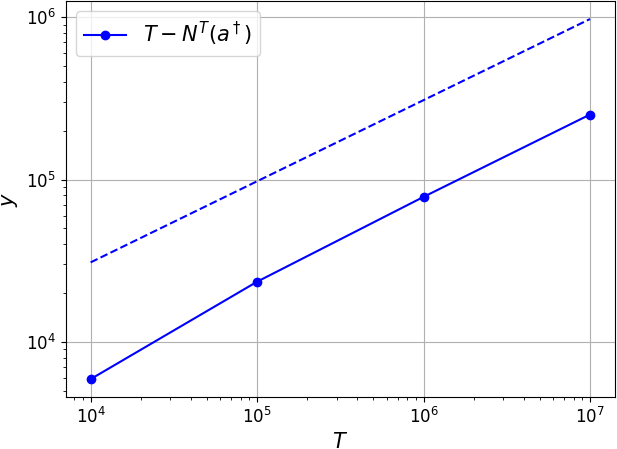}  
  \caption{Number of rounds with $a^t\neq a^\dagger$. The dashed line is the theoretical upper bound.}
  \label{fig:TC_TIA_N}
\end{subfigure}
\hfill
\begin{subfigure}{0.325\textwidth}
  \centering
  % include second image
  \includegraphics[width=1\textwidth]{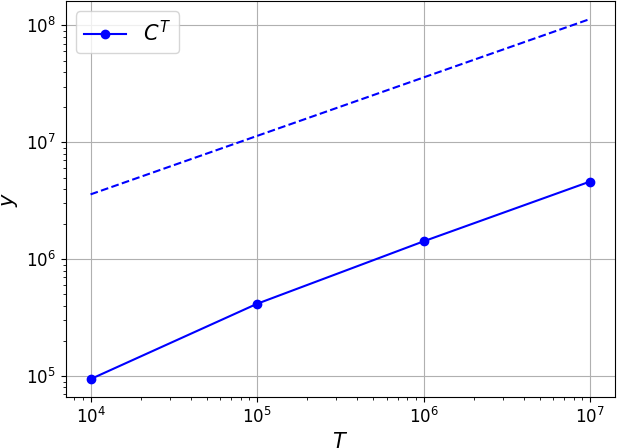}  
  \caption{The cumulative design cost. The dashed line is the theoretical upper bound.}
  \label{fig:TC_TIA_cost}
\end{subfigure}
\hfill
\begin{subfigure}{0.305\textwidth}
  \centering
  % include second image
  \includegraphics[width=1\textwidth]{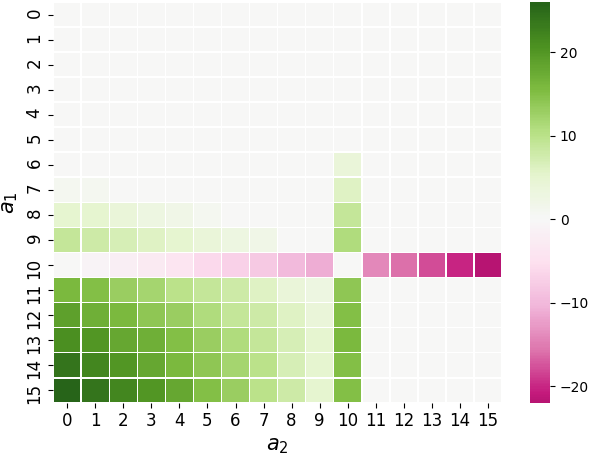}  
  \caption{Loss change $\ell_1(a)-\ell_1^o(a)$. When $a=a^\dagger=(10,10)$, the loss is unchanged.}
  \label{fig:TC_loss_change}
\end{subfigure}
\end{subfigure}
\caption{Interior design (Algorithm~\ref{alg:interior_design}) on Tragedy of the Commons.}
\label{fig:TC_TIA}
\vspace{-0.5cm}
\end{figure*}

\begin{table*}[htb]
%       \begin{minipage}[t]{0.3\linewidth}
        \begin{minipage}{0.33\linewidth} %changed <<<<<<<<<<<<
        \begin{center}
               \begin{tabular}{cc|c|c|c|}
  & \multicolumn{1}{c}{} & \multicolumn{2}{c}{} \\
  & \multicolumn{1}{c}{} & \multicolumn{1}{c}{mum}  & \multicolumn{1}{c}{fink}   \\\cline{3-4}
            & mum & $2,2$ & $5,1$  \\ \cline{3-4}
  & fink & $1,5$ & $4,4$  \\\cline{3-4}
\end{tabular}
\vspace{0.4cm}
\caption{The original loss $\ell^o$ of PD.}
\label{tab:PD_original}
\end{center}
        \end{minipage}% 
        \hfill
    %% remove blank lines <<<<<<<<<<<<<<<<
%       \begin{minipage}[t]{0.7\linewidth}
        \begin{minipage}{0.33\linewidth} %changed <<<<<<<<<<<
                \begin{center}
                \begin{tabular}{cc|c|c|c|}
  & \multicolumn{1}{c}{} & \multicolumn{2}{c}{} \\
  & \multicolumn{1}{c}{} & \multicolumn{1}{c}{mum}  & \multicolumn{1}{c}{fink}   \\\cline{3-4}
            & mum & $2,2$ & $1.5,2.5$  \\ \cline{3-4}
  & fink & $2.5,1.5$ & $4,4$  \\\cline{3-4}
\end{tabular}
\vspace{0.4cm}
\caption{The redesigned loss $\ell$ of PD.}
\label{tab:PD_attacked}
\end{center}
\end{minipage}
\hfill
\begin{minipage}{0.33\linewidth}
\begin{center}
\begin{tabular}{cc|c|c|c|}
  & \multicolumn{1}{c}{} & \multicolumn{3}{c}{} \\
  & \multicolumn{1}{c}{} & \multicolumn{1}{c}{$R$}  & \multicolumn{1}{c}{$P$}  & \multicolumn{1}{c}{$S$} \\\cline{3-5}
            & $R$ & $0,0$ & $1,-1$ & $-1,1$ \\ \cline{3-5}
  & $P$ & $-1,1$ & $0,0$ & $1,-1$ \\\cline{3-5}
            & $S$ & $1,-1$ & $-1,1$ & $0,0$ \\\cline{3-5}\
\end{tabular}
\caption{The original loss $\ell^o$ of RPS.}
\label{tab:RPS}
\end{center}
\end{minipage}
\vspace{-0.4cm}
    \end{table*} 

\subsection{Tragedy of the Commons (TC)}
Our second example is the tragedy of the commons (TC). There are $M=2$ farmers who share the same pasture to graze sheep. Each farmer $i$ is allowed to graze at most 15 sheep, i.e., the action space is $\A_i=\{0,1,...,15\}$. The more sheep are grazed, the less well fed they are, and thus less price on market. We assume the price of each sheep is $p(a)=\sqrt{30-\sum_{i=1}^2 a_i}$, where $a_i$ is the number of sheep that farmer $i$ grazes. The loss function of farmer $i$ is then $\ell_i^o(a)=-p(a)a_i$, 
i.e. negating the total price of the sheep that farmer $i$ owns. 
The Nash equilibrium strategy of this game is that every farmer grazes $a_i^*=\frac{60}{2M+1}=12$ sheep, and the resulting price of a sheep is $p(a^*)=\sqrt{6}$. 

It is well-known that this Nash equilibrium is suboptimal.
Instead, the designer hopes to maximize social welfare:
$$
\sqrt{30-\sum_{i=1}^2 a_i}\left( \sum_{i=1}^2 a_i\right),$$ 
which is achieved when $\sum_{i=1}^2 a_i=20$. Moreover, to promote equity the designer desires that the two farmers each graze the same number $20/M=10$ of sheep. 
Thus the designer has a target action profile $a_i^\dagger=10, \forall i$. 
Note that the original loss function takes value in $[-15\sqrt{15}, 0]$, while the loss of the target profile is $\ell_i^o(a^\dagger)=-10\sqrt{10}$, thus this is the interior design scenario, and the designer could apply Algorithm~\ref{alg:interior_design} to produce a new game $\ell$. 
Due to the large number of entries, we only visualize the difference $\ell_1(a)-\ell^o_1(a)$ for player 1 in Figure~\ref{fig:TC_loss_change}; the other player is the same.
We observe three patterns of loss change. For most $a$'s, e.g., $a_1\le 6$ or $a_2\ge 11$, the original loss $\ell_1^o(a)$ is already sufficiently large and satisfies the dominance gap in~\lemref{lem:post-attack-cost}, thus the designer leaves the loss unchanged. For those $a$'s where $a_1^\dagger=10$, the designer reduces the loss to make the target action more profitable. For those $a$'s close to the bottom left ($a_1> a_1^\dagger$ and $a_2\le 10$), the designer increases the loss to enforce the dominance gap $(1-\frac{1}{M})\rho$.

We simulated play for $T=10^4,10^5, 10^6$ and $10^7$ and show the results in Figure~\ref{fig:TC_TIA}. 
Again the game redesign is successful: the figures confirm $T- O(\sqrt{T})$ target action play 
and $O(\sqrt{T})$ cumulative design cost.
Numerically,
for the four $T$'s the designer enforces $a^\dagger$ in 41\%, 77\%, 92\%, and 98\% of rounds,
and the per-round design costs are 9.4, 4.2, 1.4, and 0.5, respectively.

\subsection{Prisoner's  Dilemma (PD)}
Out third example is the prisoner's dilemma (PD). There are two prisoners, each can stay mum or fink. The original loss function $\ell^o$ is given in Table~\ref{tab:PD_original}. The Nash equilibrium strategy of this game is that both prisoners fink. 
Suppose a mafia designer hopes to force $a^\dagger=$(mum, mum) by sabotaging the losses. Note that $\forall i, \ell_i^o(a^\dagger)=2$, which lies in the interior of the loss range $\L=[1, 5]$. Therefore, this is again an interior design scenario, and the designer can apply Algorithm~\ref{alg:interior_design}. In Table~\ref{tab:PD_attacked} we show the redesigned game $\ell$. 
Note that when both prisoners stay mum or both fink, the designer does not change the loss. On the other hand, when one prisoner stays mum and the other finks, the designer reduces the loss for the mum prisoner and increases the loss for the betrayer.

We simulated plays for $T=10^4, 10^5, 10^6$, and $10^7$, respectively.
In Figure~\ref{fig:PD_TIA} we plot the number of non-target action selections $T-N^T(a^\dagger)$ and the cumulative design cost $C^T$.  Both grow sublinearly.  
The designer enforces $a^\dagger$ in 85\%, 94\%, 98\%, and 99\% of rounds, and the per-round design costs are
0.71, 0.28, 0.09, and 0.03, respectively.

\begin{figure}[t]
\centering
\begin{subfigure}{0.48\linewidth}
  \centering
  \includegraphics[width=1\textwidth]{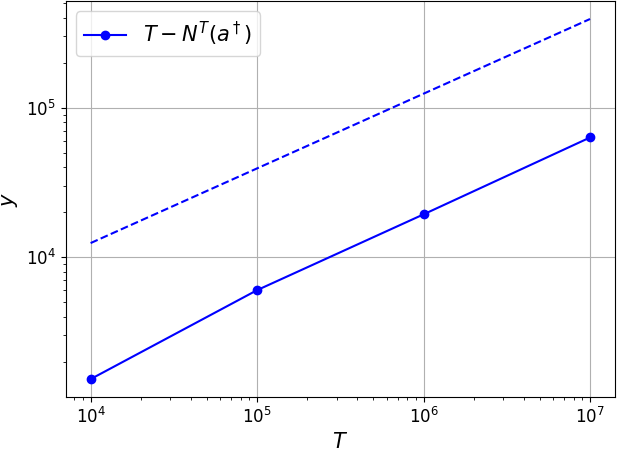}  
  \caption{Number of rounds with $a^t\neq a^\dagger$. The dashed line is the theoretical upper bound.}
  \label{fig:PD_TIA_N}
\end{subfigure}
\hfill
\begin{subfigure}{0.48\linewidth}
  \centering
  \includegraphics[width=1\textwidth]{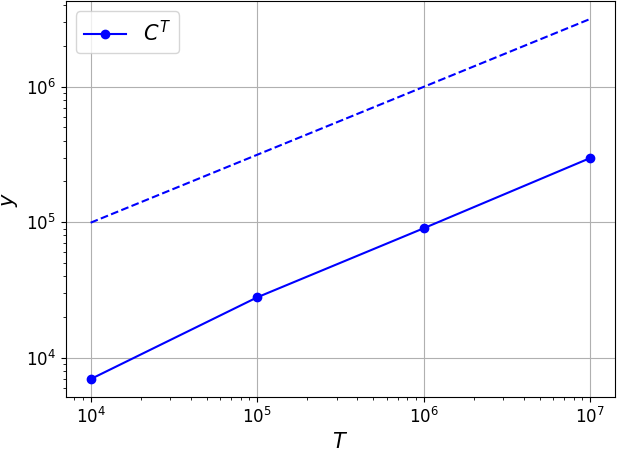}  
  \caption{The cumulative design cost. The dashed line is the theoretical upper bound.}
  \label{fig:PD_TIA_cost}
\end{subfigure}
\caption{Interior design on Prisoner's Dilemma.}
\label{fig:PD_TIA}
\vspace{-0.4cm}
\end{figure}

\begin{table*}[t!]
\hspace{-1cm}
\centering
\begin{minipage}[c]{0.33\textwidth}
\begin{tabular}{cc|c|c|c|}
  & \multicolumn{1}{c}{} & \multicolumn{3}{c}{} \\
  & \multicolumn{1}{c}{} & \multicolumn{1}{c}{$R$}  & \multicolumn{1}{c}{$P$}  & \multicolumn{1}{c}{$S$} \\\cline{3-5}
            & $R$ & $-0.5,0.5$ & $0,0$ & $-0.5,0.5$ \\ \cline{3-5}
  & $P$ & $0,0$ & $0.5,-0.5$ & $0,0$ \\\cline{3-5}
            & $S$ & $0,0$ & $0.5,-0.5$ & $0,0$ \\\cline{3-5}\
\end{tabular}
\caption*{(a)  $\ell^t (t=1)$.}
\label{tab:RPS_loss_t1}
\end{minipage}
\hspace{-0.9cm}
\begin{minipage}[c]{0.33\textwidth}
\begin{tabular}{cc|c|c|c|}
  & \multicolumn{1}{c}{} & \multicolumn{3}{c}{} \\
  & \multicolumn{1}{c}{} & \multicolumn{1}{c}{$R$}  & \multicolumn{1}{c}{$P$}  & \multicolumn{1}{c}{$S$} \\\cline{3-5}
            & $R$ & $0.62,-0.62$ & $0.75,-0.75$ & $0.62,-0.62$ \\ \cline{3-5}
  & $P$ & $0.75,-0.75$ & $0.87,-0.87$ & $0.75,-0.75$ \\\cline{3-5}
            & $S$ & $0.75,-0.75$ & $0.87,-0.87$ & $0.75,-0.75$ \\\cline{3-5}\
\end{tabular}
\caption*{\hspace{0.7cm} (b)  $\ell^t (t=10^3)$.}
\label{tab:RPS_loss_t2}
\end{minipage}
\begin{minipage}[c]{0.33\textwidth}
\begin{tabular}{cc|c|c|c|}
  & \multicolumn{1}{c}{} & \multicolumn{3}{c}{} \\
  & \multicolumn{1}{c}{} & \multicolumn{1}{c}{$R$}  & \multicolumn{1}{c}{$P$}  & \multicolumn{1}{c}{$S$} \\\cline{3-5}
            & $R$ & $0.94,-0.94$ & $0.96,-0.96$ & $0.94,-0.94$ \\ \cline{3-5}
  & $P$ & $0.96,-0.96$ & $0.98,-0.98$ & $0.96,-0.96$ \\\cline{3-5}
            & $S$ & $0.96,-0.96$ & $0.98,-0.98$ & $0.96,-0.96$ \\\cline{3-5}\
\end{tabular}
\caption*{\hspace{0.9cm}(c) $\ell^t (t=10^7)$.}
\label{tab:RPS_loss_tinf}
\end{minipage}
\caption{The redesigned RPS games $\ell^t$ for selected $t$ (with $\epsilon=0.3$).  Note the target entry $a^\dagger=(R,P)$ converges toward $(1,-1)$.}
\label{tab:RPS_loss_t}
\vspace{-0.8cm}
\end{table*}

\begin{table*}[t!]
\hspace{-1cm}
\begin{minipage}[c]{0.33\textwidth}
\centering
\begin{tabular}{cc|c|c|c|}
  & \multicolumn{1}{c}{} & \multicolumn{3}{c}{} \\
  & \multicolumn{1}{c}{} & \multicolumn{1}{c}{$R$}  & \multicolumn{1}{c}{$P$}  & \multicolumn{1}{c}{$S$} \\\cline{3-5}
            & $R$ & $1,1$ & $1,1$ & $-1,1$ \\ \cline{3-5}
  & $P$ & $-1,-1$ & $1,-1$ & $-1,-1$ \\\cline{3-5}
            & $S$ & $-1,1$ & $-1,-1$ & $-1,-1$ \\\cline{3-5}\
\end{tabular}
\caption*{\hspace{0.5cm}(a)  $\hat\ell^t (t=1)$.}
\label{tab:RPS_loss_t1_discrete}
\end{minipage}
\begin{minipage}[c]{0.33\textwidth}
\centering
\begin{tabular}{cc|c|c|c|}
  & \multicolumn{1}{c}{} & \multicolumn{3}{c}{} \\
  & \multicolumn{1}{c}{} & \multicolumn{1}{c}{$R$}  & \multicolumn{1}{c}{$P$}  & \multicolumn{1}{c}{$S$} \\\cline{3-5}
            & $R$ & $1,-1$ & $1,1$ & $-1,-1$ \\ \cline{3-5}
  & $P$ & $1,-1$ & $1,-1$ & $1,-1$ \\\cline{3-5}
            & $S$ & $1,-1$ & $1,-1$ & $1,1$ \\\cline{3-5}\
\end{tabular}
\caption*{\hspace{1cm}(b)  $\hat\ell^t (t=10^3)$.}
\label{tab:RPS_loss_t2_discrete}
\end{minipage}
\begin{minipage}[c]{0.33\textwidth}
\centering
\begin{tabular}{cc|c|c|c|}
  & \multicolumn{1}{c}{} & \multicolumn{3}{c}{} \\
  & \multicolumn{1}{c}{} & \multicolumn{1}{c}{$R$}  & \multicolumn{1}{c}{$P$}  & \multicolumn{1}{c}{$S$} \\\cline{3-5}
            & $R$ & $1,-1$ & $1,-1$ & $1,-1$ \\ \cline{3-5}
  & $P$ & $1,-1$ & $1,-1$ & $1,-1$ \\\cline{3-5}
            & $S$ & $1,-1$ & $1,-1$ & $1,-1$ \\\cline{3-5}\
\end{tabular}
\caption*{\hspace{1cm}(c) $\hat\ell^t (t=10^7)$.}
\label{tab:RPS_loss_tinf_discrete}
\end{minipage}
\caption{Instantiation of discrete design on the same games as in Table~\ref{tab:RPS_loss_t}.  The redesigned loss lies in $\L=\{-1,0,1\}$.}
\label{tab:RPS_loss_t_discrete}
\end{table*}

\begin{figure*}[t!]
\centering
\begin{minipage}{0.48\textwidth}
\begin{subfigure}{1\textwidth}
\begin{subfigure}{0.48\textwidth}
  % include first image
  \includegraphics[width=1\textwidth]{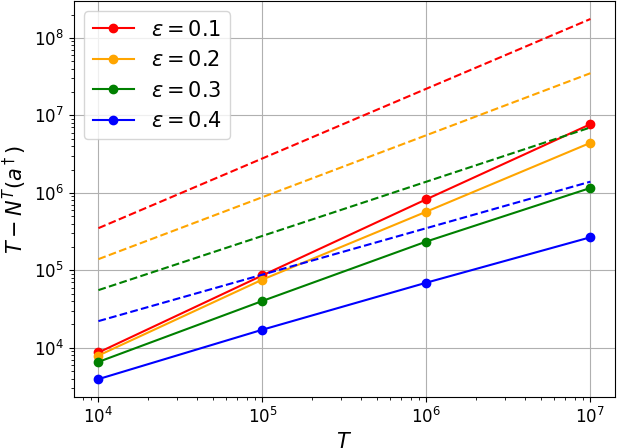}  
  \caption{Number of rounds with $a^t\neq a^\dagger$. The dashed lines are the theoretical upper bound.}
  \label{fig:RPS_TVA_N}
\end{subfigure}
\hfill
\begin{subfigure}{0.48\textwidth}
  % include second image
  \includegraphics[width=1\textwidth]{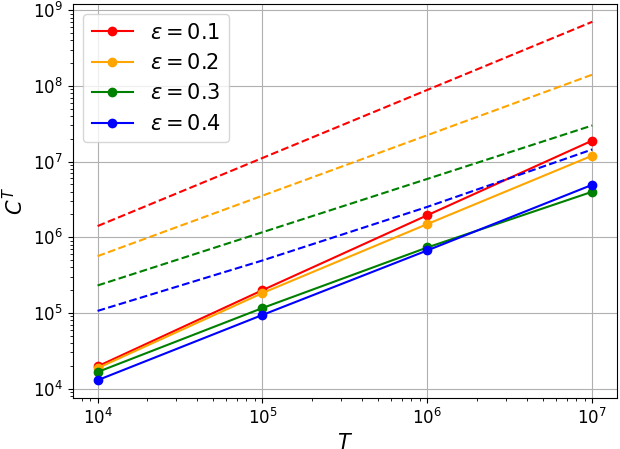}  
  \caption{The cumulative design cost. The dashed lines are the theoretical upper bound.}
  \label{fig:RPS_TVA_cost}
\end{subfigure}
\end{subfigure}
\caption{Boundary design on RPS.}
\label{fig:RPS_TVA}
\end{minipage}
\hfill
\begin{minipage}{0.48\textwidth}
\begin{subfigure}{1\textwidth}
\begin{subfigure}{0.48\textwidth}
  \centering
  % include first image
  \includegraphics[width=1\textwidth]{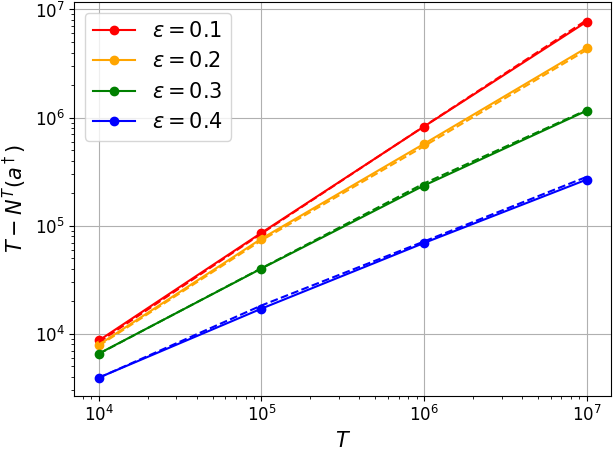}  
  \caption{Number of rounds  $a^t\neq a^\dagger$} 
  \label{fig:RPS_TVA_prob_N}
\end{subfigure}
\hfill
\begin{subfigure}{0.48\textwidth}
  \centering
  % include second image
  \includegraphics[width=1\textwidth]{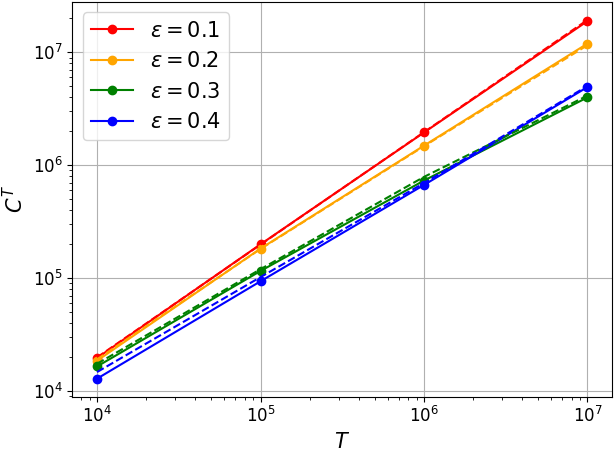}  
  \caption{Cumulative design cost $C^T$}
  \label{fig:RPS_TVA_prob_cost}
\end{subfigure}
\end{subfigure}
\caption{Discrete redesign for $a^\dagger=(R,P)$ with natural loss values in $\L$.  The dashed lines are the corresponding boundary design with unnatural loss values.}
\label{fig:RPS_TVA_prob}
\end{minipage}
\end{figure*}

\subsection{Rock-Paper-Scissors (RPS)}
While redesigning RPS does not have a natural motivation, it serves as a clear example on how boundary design and discrete design can be carried out on more socially-relevant games. 
The original game $\ell^o$ is in Table~\ref{tab:RPS}.

\textbf{Boundary Design.}
Suppose the designer target action profile is $a^\dagger=(R, P)$, namely making the row player play Rock while the column player play Paper. 
Because $\ell^o(a^\dagger)=(1,-1)$ hits the boundary of loss range $\tilde \L=[-1,1]$, 
the designer can use the Boundary Design Algorithm~\ref{alg:boundary_design}. 
For simplicity we choose $v$ with $v_i=\frac{U+L}{2}, \forall i$. Because in RPS $U=-L=1$, this choice of $v$ also preserves the zero-sum property. 
Table~\ref{tab:RPS_loss_t} shows the redesigned games at $t=1, 10^3$ and $10^7$ under $\epsilon=0.3$. Note that the designer maintains the zero-sum property of the games. Also note that the redesigned loss function always guarantees strict dominance of $a^\dagger$ for all $t$, but the dominance gap decreases as $t$ grows. Finally, the loss of the target action $a^\dagger=(R, P)$ converges to the original loss $\ell^o(a^\dagger)=(1, -1)$ asymptotically, thus the designer incurs diminishing design cost.

We ran Algorithm~\ref{alg:boundary_design} under four values of $\epsilon=0.1,0.2,0.3,0.4$, resulting in four game sequence $\ell^t$.
For each $\epsilon$ we simulated game play for $T=10^4, 10^5, 10^6$ and $10^7$. In Figure~\ref{fig:RPS_TVA}\subref{fig:RPS_TVA_N}, we show $T-N^T(a^\dagger)$ under different $\epsilon$ (solid lines). We also show the theoretical upper bounds of~\thmref{thm:attack_version_2} (dashed lines) for comparison. 
In Figure~\ref{fig:RPS_TVA}\subref{fig:RPS_TVA_cost}, we show the cumulative design cost $C^T$ under different $\epsilon$.  
The theoretical values indeed upper bound our empirical results. Furthermore, all non-target action counts and all cumulative design costs grow only sublinearly. 
As an example,
under $\epsilon=0.3$ for the four $T$'s the designer forces $a^\dagger$ in 34\%, 60\%, 76\%, and 88\% rounds, respectively.
The per-round design costs are 1.7, 1.2, 0.73 and 0.40, respectively. The results are similar for the other $\epsilon$'s. 
Choosing the best $\epsilon$ and $v$ is left as future work. We note that empirically the cumulative design cost achieves the minimum at some $\epsilon\in(0.3, 0.4)$ while~\thmref{thm:attack_version_2} suggests that the minimum cost is at $\epsilon^*=0.25$ instead. We investigate this inconsistency in the appendix~\ref{sec:optimal_eps}.

\textbf{Discrete Design.}
In our second experiment on RPS, we compare the performance of discrete design (Algorithm~\ref{alg:discrete_design}) with the deterministic boundary design (Algorithm~\ref{alg:boundary_design}).
Again, the target action profile is $a^\dagger=(R,P)$.
Recall the purpose of discrete design is to only use natural game loss values, in the RPS case $\L=\{-1,0,1\}$ instead of unnatural real values in the relaxed $\tilde\L=[-1,1]$, to make the redesign ``less detectable'' by players.  
We hope to show that discrete design does not lose much potency even with this value restriction.
Figure~\ref{fig:RPS_TVA_prob} shows this is indeed the case.  Discrete design performance nearly matches boundary design.
For example when $\epsilon=0.3$, for the four $T$'s discrete design enforces $a^\dagger$ 35\%, 59\%,75\% and 88\% of the time.
The per-round design costs are 1.7, 1.2, 0.79, and 0.41, respectively.
Overall, discrete design does not lose much performance and may be preferred by designers. Table~\ref{tab:RPS_loss_t_discrete} shows the redesigned ``random'' games at $t=1, 10^3$ and $10^7$ under $\epsilon=0.3$. Note that the loss lies in the natural range $\L=\{-1,0,1\}$. Also note that the loss function converges to be a constant function that takes the target loss value $\ell^o(a^\dagger)$. Finally, we point out that in general, the discrete design does not preserve the zero-sum property.

\section{Conclusion and Future Work}
In this paper, we studied the problem of game redesign where the players apply no-regret algorithms to play the game. We show that a designer can force all players to play a target action profile in $T-o(T)$ rounds while incurring only $o(T)$ cumulative design cost. We develop redesign algorithms for both the interior and the boundary target loss scenarios. Experiments on four game examples demonstrate the performance of our redesign algorithms.
Future work could study 
defense mechanisms 
to mitigate the effect of game redesign when the designer is malicious; or when the designer is one of the players 
with more knowledge than other players and willing to intentionally lose.
%===
%There are some future research questions: (1) Note that the players in our current setup are no-regret learners and they do not infer the policy or behavior of the other players when choosing actions. Therefore, they are not real game-theoretic players. Future work should generalize to more sophisticated multi-agent learning scenarios where the players adopt real game-theoretic behaviors, such as in stochastic Markov games. (2) Our game redesign can be exploited by malicious designers to achieve nefarious goals. In this case, how can we design defense mechanisms to mitigate the effect of game redesign? (3) In our paper, we consider an external designer who is oblivious to the game players.  Another possibility is that the designer is also a player of the game (i.e., an internal designer), which the players are aware of. However, the designer may have more knowledge and power than existing game players. Future work should study how to model the internal designer and develop efficient redesign algorithms.
\newpage
\bibliography{ref}
\bibliographystyle{abbrv}

\newpage
\begin{appendix}
\clearpage
% \section{Convergence Results of No-regret Learners}
% \label{app:convergence}
% \begin{remark}
% Let $\pi^t=(\pi_1^t,...\pi_M^t)$ be the joint strategy of all players at round $t$. Then for two-player ($M=2$) zero-sum games (i.e, $\sum_i \ell_i^o(a)=0, \forall a$), no-regret learners guarantee $\E{}{\frac{1}{T}\sum_t \pi^t}$ converges to some Nash equilibrium.
% \end{remark}
% 
% \begin{remark}
% Let $\pi^t=(\pi_1^t,...\pi_M^t)$ be the joint strategy of all players at round $t$. Consider the following empirical distribution $D_T$: first draw $t\sim U([1:T])$, where $U$ is the uniform distribution, and then each player $i$ follows strategy $\pi_i^t$. For multi-player ($M\ge 2$) general-sum games, no-regret learners guarantee $\E{}{D_T}$ converges to some coarse correlated equilibrium.
% \end{remark}

\section{Additional Proofs}

\lemPostAttackCost*
\begin{proof}
The redesigned game~\eqref{eq:interior_design} is given by
\begin{equation}\label{eq:interior_design_copy}
\forall i, a,  \ell_i(a)=\left\{
\begin{array}{ll}
\ell_i^o(a^\dagger)- (1-\frac{d(a)}{M})\rho & \mbox{ if } a_i= a_i^\dagger, \\
\ell_i^o(a^\dagger)+\frac{d(a)}{M}\rho& \mbox{ if } a_i\neq a_i^\dagger,
\end{array}
\right.
\end{equation}
where $d(a)=\sum_{j=1}^M \ind{a_j=a_j^\dagger}$.

\begin{enumerate}[leftmargin=*]

\item

Both branches of $\ell_i(a)$ are lower bounded by $L$:
\begin{equation}
\begin{aligned}
\ell_i^o(a^\dagger)- (1-\frac{d(a)}{M})\rho &\ge \ell_i^o(a^\dagger)-\rho\ge L.
\end{aligned}
\end{equation}
\begin{equation}
\begin{aligned}
\ell_i^o(a^\dagger)+\frac{d(a)}{M}\rho &\ge \ell_i^o(a^\dagger)\ge L.
\end{aligned}
\end{equation}
Both branches are upper bounded by $U$:
\begin{equation}
\begin{aligned}
\ell_i^o(a^\dagger)- (1-\frac{d(a)}{M})\rho &\le \ell_i^o(a^\dagger)\le U.
\end{aligned}
\end{equation}
\begin{equation}
\begin{aligned}
\ell_i^o(a^\dagger)+\frac{d(a)}{M}\rho &\le \ell_i^o(a^\dagger) + \rho\le U.
\end{aligned}
\end{equation}
Therefore, $\ell_i(a) \in [L, U]=\tilde \L$.

\item
Fix $i\in[M]$.
$\forall a_{-i}$, let $a=(a_i, a_{-i})$ for some $a_i\neq a_i^\dagger$, and $b=(a_i^\dagger, a_{-i})$, then we have $d(b)=d(a)+1$, thus 
\begin{equation}
\begin{aligned}
&\ell_i(a)-\ell_i(b)=\ell_i^o(a^\dagger)+\frac{d(a)}{M}\rho-\ell_i^o(a^\dagger)+(1-\frac{d(b)}{M})\rho\\
&=(1-\frac{1}{M})\rho.
\end{aligned}
\end{equation}
Therefore, for player $i$ the target action $a_i^\dagger$ strictly dominates any other actions by $(1-\frac{1}{M})\rho$.

\item
When $a=a^\dagger$, we have $d(a)=M$, thus by our design, we have $\forall i$,
\begin{equation}
\begin{aligned}
\ell_i(a^\dagger) &= \ell^o_i(a^\dagger)-(1-\frac{d(a)}{M})\rho=\ell^o_i(a^\dagger)-(1-\frac{M}{M})\rho=\ell^o_i(a^\dagger).
\end{aligned}
\end{equation}

\item
Fix $a$, we sum over all players to obtain
\begin{equation}
\begin{aligned}
&\sum_{i=1}^M \ell_i(a)=\sum_{i:a_i=a_i^\dagger} \left(\ell_i^o(a^\dagger)-(1-\frac{d(a)}{M})\rho\right)+\\
&\sum_{i:a_i\neq a_i^\dagger} \left(\ell_i^o(a^\dagger)+\frac{d(a)}{M}\rho\right)\\
&=\sum_i \ell_i^o(a^\dagger)-d(a)(1-\frac{d(a)}{M})\rho+(M-d(a))\frac{d(a)}{M}\rho\\
&=\sum_{i=1}^M \ell_i^o(a^\dagger)=0.
\end{aligned}
\end{equation}
\end{enumerate}
\end{proof}

\lemPostAttackCostVersionTwo*
\begin{proof}
The redesigned game~\eqref{eq:boundary_design} is given by
\begin{equation}\label{eq:boundary_design_copy}
\ell^t= w_t \underline\ell + (1-w_t)\overline \ell
\end{equation}
where 
\begin{equation}\label{eq:wt_copy}
w_t=t^{\alpha+\epsilon-1}
\end{equation}

\begin{enumerate}[leftmargin=*]

\item 
Note that $\underline \ell$ is valid, as we have proved in~\lemref{lem:post-attack-cost} property~\ref{lem:post-attack-cost-prop1}, thus $\underline \ell\in[L, U]$. Also note that $\overline \ell\in [L, U]$. Therefore, $\ell^t=w_t \underline \ell+(1- w_t)\overline \ell\in [L, U]$.

\item

$\forall i$ and $\forall a_{-i}$, let $a=(a_i, a_{-i})$ for some $a_i\neq a_i^\dagger$, and let $b=(a_i^\dagger, a_{-i})$, then according to~\lemref{lem:post-attack-cost} property~\ref{lem:post-attack-cost-prop2}, we have
\begin{equation}
\underline \ell(a)-\underline \ell(b)=(1-\frac{1}{M})\rho.
\end{equation}
Therefore, we have
\begin{equation}
\begin{aligned}
&\ell^t(a)-\ell^t(b)=(1- w_t)\overline \ell(a)+ w_t \underline \ell(a)-(1- w_t)\overline \ell(b)+ w_t \underline \ell(b)\\
&=(1- w_t)\ell^o(a^\dagger)+ w_t \underline \ell(a)-(1- w_t)\ell^o(a^\dagger)+ w_t \underline \ell(b)\\
&= w_t\left(\underline \ell(a)-\underline \ell(b)\right)=(1-\frac{1}{M})\rho w_t.
\end{aligned}
\end{equation}

\item 

Note that we have
\begin{equation}
\begin{aligned}
&\ell^o(a^\dagger)-\ell^t(a^\dagger)=\ell^o(a^\dagger)-\left(w_t\underline \ell(a^\dagger)+(1- w_t)\ell^o(a^\dagger)\right)\\
&=w_t\left(\ell^o(a^\dagger)-\underline \ell(a^\dagger)\right).
\end{aligned}
\end{equation}
Therefore, we have 
\begin{equation}
\begin{aligned}
&C(\ell^o, \ell^t, a^\dagger)\le \eta\|\ell^o(a^\dagger)- \ell^t(a^\dagger)\|_p\\
&= \eta w_t\|\ell^o(a^\dagger)-\underline \ell(a^\dagger)\|_p\\
&\le \eta (U-L)M^{\frac{1}{p}} w_t.
\end{aligned}
\end{equation}

\item
If the loss vector $v$ is zero-sum, then by~\lemref{lem:post-attack-cost} property~\ref{lem:post-attack-cost-prop4} $\underline \ell$ is a zero-sum game. If $\ell^o(a^\dagger)$ is also zero-sum, then we have
\begin{equation}
\begin{aligned}
&\sum_{i=1}^M \ell^t_i(a)=\sum_{i=1}^M \left(w_t\underline \ell_i(a)+(1- w_t)\ell_i^o(a^\dagger)\right)\\
&=w_t \sum_{i=1}^N\underline \ell_i(a)+(1- w_t)\sum_{i=1}^M \ell_i^o(a^\dagger)=0.
\end{aligned}
\end{equation}
\end{enumerate}
\end{proof}

\section{Exact Form of the Theoretical Upper Bounds}\label{sec:exact_form}
According to Theorem 3.4 in~\cite{bubeck2012regret}, the EXP3.P achieves expected regret bound
\begin{equation}
\E{}{R^T}\le 5.15\sqrt{TA_i\log A_i}+\sqrt{\frac{TA_i}{\log A_i}}.
\end{equation}
where $A_i=|\A_i|$ is the size of the action space of player $i$. Note that, however,~\cite{bubeck2012regret} assumes the loss takes value in [0, 1], while we assume the loss lies in $[L, U]$. Therefore, the regret bound should boost by $U-L$, i.e., we have
\begin{equation}
\forall i, \E{}{R_i^T}\le(U-L) \left(5.15\sqrt{TA_i\log A_i}+\sqrt{\frac{TA_i}{\log A_i}}\right).
\end{equation}
Plug the above regret bound into the proofs of~\thmref{thm:attack_version_01} and~\thmref{thm:attack_version_2}, we obtain the following exact form of the theoretical upper bounds.

For the interior design Algorithm~\ref{alg:interior_design}, we have
\begin{equation}\label{eq:interior_exact_bound_N}
\begin{aligned}
&T-\E{}{N^T(a^\dagger)}\le \sum_{j=1}^M\frac{M}{(M-1)\rho}\E{}{R_i^T}\\
&=\frac{M(U-L)}{(M-1)\rho} \sum_{i=1}^M\left(5.15\sqrt{TA_i\log A_i}+\sqrt{\frac{TA_i}{\log A_i}}\right)
\end{aligned}
\end{equation}
and
\begin{equation}\label{eq:interior_exact_bound_C}
\begin{aligned}
&T-\E{}{N^T(a^\dagger)}\le \eta M^{\frac{1}{p}}(U-L)\sum_{j=1}^M\frac{M}{(M-1)\rho}\E{}{R_i^T}\\
&=\frac{\eta M^{1+\frac{1}{p}}(U-L)^2}{(M-1)\rho} \sum_{i=1}^M\left(5.15\sqrt{TA_i\log A_i}+\sqrt{\frac{TA_i}{\log A_i}}\right)
\end{aligned}
\end{equation}

For the boundary design Algorithm~\ref{alg:boundary_design}, we have
\begin{equation}\label{eq:boundary_exact_bound_N}
\begin{aligned}
&T-\E{}{N^T(a^\dagger)}\le \sum_{i=1}^M \left(\frac{M}{(M-1)\rho}\E{}{R_i^T} T^{1-\alpha-\epsilon}+\frac{1}{\alpha+\epsilon}T^{1-\alpha-\epsilon}\right)\\
&=\left(\frac{M(U-L)}{(M-1)\rho} \sum_{i=1}^M\left(5.15\sqrt{TA_i\log A_i}+\sqrt{\frac{TA_i}{\log A_i}}\right)+\frac{M}{\alpha+\epsilon}\right)T^{1-\alpha-\epsilon}.
\end{aligned}
\end{equation}
and
\begin{equation}\label{eq:boundary_exact_bound_C}
\begin{aligned}
&\E{}{C^T}\le  \eta (U-L)M^{\frac{1}{p}}(T-\E{}{N^T(a^\dagger)})+\eta (U-L)M^{\frac{1}{p}}\sum_{t=1}^T w_t\\
&\le  \eta (U-L)M^{\frac{1}{p}}\times \text{\eqref{eq:boundary_exact_bound_N}}+\frac{\eta (U-L)M^{\frac{1}{p}}}{\alpha+\epsilon}T^{\alpha+\epsilon}.
\end{aligned}
\end{equation}

\section{Minimum Cumulative Design Cost}
\label{sec:optimal_eps}

\thmref{thm:attack_version_2} suggests that the minimum cost is achieved at $\epsilon^*=\frac{1-\alpha}{2}=0.25$, while Figure~\ref{fig:RPS_TVA}\subref{fig:RPS_TVA_cost} implies that the cost is minimum at some $\epsilon \in (0.3, 0.4)$. We believe the inconsistency is due to not large enough horizon $T$. We now experiment with slightly larger $T$ for the RPS game with $a^\dagger=(R, P)$. Specifically, we let $T=10^6,10^7, 10^8$ and $\epsilon=0.1, 0.2, 0.25, 0.3,0.4$. In Figure~\ref{fig:RPS_TVA_NC}, we plot $\log (T-N^T(a^\dagger))$ against $\log C^T$ and we marked out the corresponding $\epsilon$ values on the curve. Note that for different $T$, the pattern remains the same -- as $\epsilon$ grows, $\log(T-\log N^T(a^\dagger))$ decreases monotonically, while $\log C^T$ first reduces and then increases. We also note that as $T$ becomes larger, the $\epsilon$ with the minimum cumulative design cost becomes closer to $\epsilon^*=0.25$. We anticipate that as $T$ grows even larger (e.g., $10^{10}$), the cumulative design cost will achieve the minimum at exactly $\epsilon^*=0.25$.

\begin{figure}[t!]
\centering
  % include first image
  \includegraphics[width=0.35\textwidth]{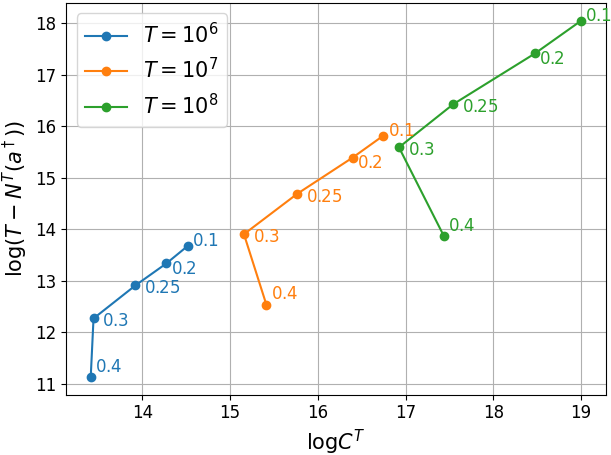}  
  \caption{Number of rounds with $a^t\neq a^\dagger$. The dashed lines are the theoretical upper bound.}
  \label{fig:RPS_TVA_NC}
\end{figure}

\end{appendix}

\end{document}